\documentclass{article}
\usepackage[utf8]{inputenc}
\usepackage{amsmath}
\usepackage{amsthm}
\usepackage{bm}
\usepackage{hyperref}
\usepackage{cleveref}
\usepackage{cite}
\hypersetup{
    colorlinks=true,
    linkcolor=blue,
    filecolor=magenta,      
    urlcolor=cyan,
}
\usepackage{graphicx}
\usepackage{mathtools}
\usepackage{setspace}
\usepackage[overload]{empheq}
\usepackage{authblk}
\newtheorem{theorem}{Theorem}[section]
\newtheorem{prop}{Proposition}
\DeclareMathOperator*{\argmax}{argmax}
\usepackage[a4paper,top=3cm,bottom=2cm,left=3cm,right=3cm,marginparwidth=1.75cm]{geometry}

\title{Mean-Field Game Analysis of SIR Model with Social Distancing}
    
\date{\today}

\author{\small Samuel Cho}

\affil{\footnotesize Program for Quantitative and Computational Biology, Princeton University}

\begin{document}
\maketitle
\begin{abstract}
The current COVID-19 pandemic has proven that proper control and prevention of infectious disease require creating and enforcing the appropriate public policies. One critical policy imposed by the policymakers is encouraging the population to practice social distancing (i.e. controlling the contact rate among the population). Here we pose a mean-field game model of individuals each choosing a dynamic strategy of making contacts, given the trade-off of gaining utility but also risking infection from additional contacts. We compute and compare the mean-field equilibrium (MFE) strategy, which assumes each individual acting selfishly to maximize its own utility, to the socially optimal strategy, which maximizes the total utility of the population. We prove that the optimal decision of the infected is always to make more contacts than the level at which it would be socially optimal, which reinforces the important role of public policy to reduce contacts of the infected (e.g. quarantining, sick paid leave). Additionally, we include cost to incentivize people to change strategies, when computing the socially optimal strategies. We find that with this cost, policies reducing contacts of the infected should be further enforced after the peak of the epidemic has passed. Lastly, we compute the price of anarchy (PoA) of this system, to understand the conditions under which large discrepancies between the MFE and socially optimal strategies arise, which is when intervening public policy would be most effective. 
\end{abstract}

\section{Introduction}
The current COVID-19 pandemic has been evidence that posing a public policy to fight such a crisis is an extremely difficult and interdisciplinary task. A major component of such policies is urging people to practice social distancing. By reducing interpersonal contacts, the spread of infection can be slowed down. However, it is uncertain how to incentivize people to practice social distancing, when there are clearly numerous benefits to making contacts, such as working for income, general desire for freedom and social relationships.\\
\\
Given this trade-off between the additional utility from making more contacts, and the additional chance of infection from those contacts, we will compute and compare 1) the selfish strategies, in which individuals make contacts to optimize only their own utilities and 2) socially optimal strategies, in which individuals make contacts which optimize the utility of the total population. \\
\\
The approach of focusing on the economic causes and epidemiological consequences is referred to as economic epidemiology\cite{perrings2014merging}. Recent work has focused on treating economic factors behind contact and mixing decisions as part of the disease transmission mechanism\cite{gersovitz2003infectious,gersovitz2004economical,barrett2007optimal,funk2009spread,funk2010modelling}. More specifically, some previous work in this field has explicitly included contact behaviors as control variables into the classical SIR model and have shown that different dynamics can emerge with adaptive behavior\cite{reluga2010game,fenichel2011adaptive,morin2013sir}. Disease dynamics has also been studied with game theory, focusing mostly on steady-state problems or those related to vaccination\cite{francis2004optimal,reluga2009sis,doncel2017mean}. Optimal control theory has also been used to study policy interventions on infectious disease dynamics\cite{lenhart2007optimal,neilan2010introduction}. Our work here builds on these previous models to pose a mean-field game problem of social distancing, explicitly modeling the feedback between the individuals and the population structure. The individual, which is susceptible (S), infected (I), or recovered (R), each chooses a dynamic contact strategy to maximize its accumulated utility over the time period.\\
\\
Mean-field games is a recently developed field of mathematics, studying dynamic game-theoretic problems of infinite number of players\cite{lasry2007mean,huang2006large}. Our model is a relatively simple formulation of the mean-field games, in which we have a deterministic game with 3 discrete states (S,I,R) in continuous time. Depending on the number of contacts each individual makes, the dynamics of the infectious population follow the SIR model, and the anticipation of this population structure influences the computation of an individual's optimal strategy. We take mean-field assumptions such that 1) the number of individuals in the population is large $(N\to\infty)$, 2) the individuals are homogeneous within each compartment and the population is well-mixed, and 3) individuals engage in symmetric interactions.
\section{Model}
\subsection{SIR dynamics}
The focus of this work is on understanding the effect of social distancing behavior and the corresponding utility trade-offs on the epidemiological dynamics, so we will keep the epidemiology as simple as possible as a baseline model. We will devote a section to including an exposed class (SEIR), but leave further expansion on the epidemiology as future work. \\
\\
We consider an SIR model without births or deaths. The recovered (R) is an absorbing state, where individuals gain permanent immunity from the disease. For now, we do not consider exposed period or asymptomatic or presymptomatic transmissions.\\
\\
Without birth or death, the population is fixed, and we let $x_z$ be the fraction of the population in compartment $z \in \{S,I,R\}$. Then the dynamics of the SIR system is
\begin{align}[left=\empheqlbrace] 
\Dot{x}_S &= -C(\cdot)\beta x_Sx_I \nonumber\\
\Dot{x}_I &= C(\cdot)\beta x_S x_I - \mu x_I \nonumber\\
\Dot{x}_R &= \mu x_I \label{eq:1}
\end{align}
$\beta$ is the likelihood that infection happens given a contact between S and I, and $\mu$ is the constant rate of recovery from infection. $C(\cdot)$ is the rate that an S individual and an I individual make contact, which is dependent on whether individuals choose to socially distance or not. Let us write the number of contacts made by a $z \in \{S,I,R\}$ individual at time $t$ as $c_z(t)$. \\
\\
It is evident that $C(\cdot)$ should increase as $c_S(t)$ and $c_I(t)$ increase, but there are two intuitive ways of writing an expression for $C(\cdot)$, analogous to frequency-dependent or density-dependent disease transmission.\\
\\
i) Frequency-dependent\\
Looking only at the transmission term (suppressing notation for dependence on $t$ for brevity),
\begin{equation}
\Dot{x}_S = -\beta\frac{c_Sx_Sc_Ix_I}{c_Sx_S + c_Ix_I + c_Rx_R}
\end{equation}
where a contact by S has probability $\frac{c_Ix_I}{c_Sx_S+c_Ix_I+c_Rx_R}$ of being with I. Therefore, the transmission term depends on the relative frequency between the contact strategies.\\
\\
ii) Density-dependent \\
Similarly the transmission term is
\begin{equation}
\Dot{x}_S = -\beta'c_Sx_Sc_Ix_I
\end{equation}
Here, the transmission term depends on not only the relative frequency, but also the absolute density of the contact strategies. If the whole population is doubling or halving contacts, the transmission should increase or decrease accordingly, and this density-dependent formulation is consistent with such intuitions. For the rest of this paper, we will take the density-dependent formulation. (Note that $\beta'$ term here is technically different from the $\beta$ term above, but for simplicity, any $\beta$ referred to for the rest of this paper is $\beta'$ of the density-dependent formulation.)

\subsection{Utility of contacts}
$c_z(t)$ is the contact strategy which is explicitly chosen by an individual. Large $c_z(t)$ means going out to work and socializing, while smaller $c_z(t)$ means refraining from those activities (i.e. social distancing). $u_z(c_z)$ denotes the utility received by a $z$ individual from making $c_z$ contacts, which is a combination of economic gains and personal well-being. The realistic conditions that we impose are the following: Starting from zero contacts, increasing the number of contacts results in increased utility, which is the baseline universal need for social interactions, as well as monetary gain from economic activities. We assume that the marginal increase in utility decreases with more contacts, until eventually, more contacts become detrimental. \\
\\
Therefore, the appropriate functional form of $u_z(c_z)$ must be a concave function with an interior maximum. One appropriate function form with interpretations of the parameters is $u_z(c_z)= (b_zc_z - c_z^2)^\gamma - a_z$ where $c_z \in [0,b_z]$\cite{fenichel2011adaptive}. $b_z$ is a parameter of how the disease impacts the marginal economic productivity of the individual. Large $b_z$ means that the individual has the choice to make more contacts, as well as receive higher marginal utility from making contacts. $a_z$ is a parameter of the baseline cost of being infected, such as an individual's general propensity to be healthy. These two parameters decompose the effect of the disease into economic cost and health cost. Some assumptions on these parameters are that $b_S=b_R>b_I$ and $a_S=a_R < a_I$, since the infected become impaired in economic productivity, as well as suffer the health cost compared to the S and R. $\gamma$ changes the concave shape of the function, and $\gamma \in (0,1]$ ensures that $u_z(c_z)$ is concave everywhere in the domain. \\
\\
From this utility function, we see that each individual has some optimal level of social contacts. The utility of a $z$ individual is maximized at $\frac{b_z}{2}$, which is each individual's optimal contact strategy in the absence of adaptive behavior in response to the risk of infectious disease.

\begin{figure}[h!]
\centering
\includegraphics[scale=0.5]{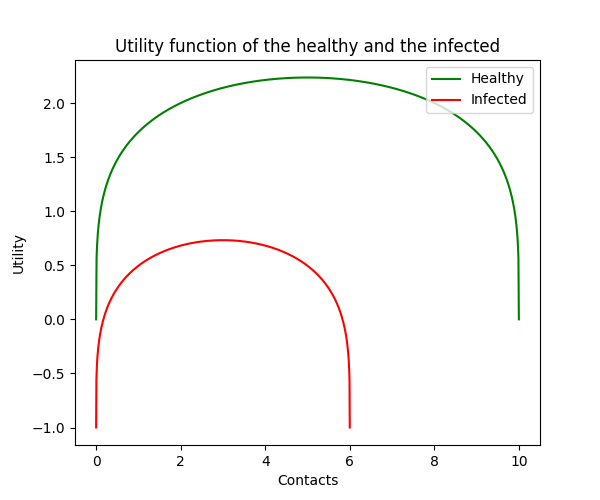} 
\caption{$u_z(c_z)$ is shown where $z \in \{S,R\}$ are healthy (green) and $z \in \{I\}$ is infected (red) and the parameters are $b_S=b_R=10$, $b_I=6$, $a_S=a_R=0$, $a_I=1$, and $\gamma=0.25$. A healthy individual gains utility from making more contacts, but eventually does not want to make more, when $c_S=5$. The infected suffers the baseline cost, but also gains some utility from making contacts, although at a lower rate compared to the healthy.}
\label{fig:util_fig}
\end{figure}

\subsection{Value function}
Over time period, $t \in [0, T]$, the total utility of an individual is the sum of utility gained at each time point. For example, a recovered individual with continuous contact decision $c(t)$ receives total utility $\int_0^T u_{R}(c(t))dt$. Typically, a discounting term, $\delta$ which discounts the future utility compared to the present utility, is included, but we set $\delta=1$ for simplicity. The general results here do not change with $\delta < 1$. \\
\\
Let us define $V_R(t)$ to be the total future utility expected by the R individual at time $t$ until $T$. From this formulation, it follows that terminal condition is $V_R(T)=0$, and 
\begin{equation}
V_R(t)=\int_t^T u_{R}(c(t))dt = u_{R}(c(t))dt + \int_{t+dt}^T u_{R}(c(t))dt = u_{R}(c(t))dt + V_R(t+dt)
\end{equation}
Similarly for S and I individuals, we define $V_S(t)$ and $V_I(t)$, which depend on the rate at which the individuals move between the SIR states. (Fig. \ref{fig:prob_fig})\\
\begin{figure}[h!]
\centering
\includegraphics[scale=0.3]{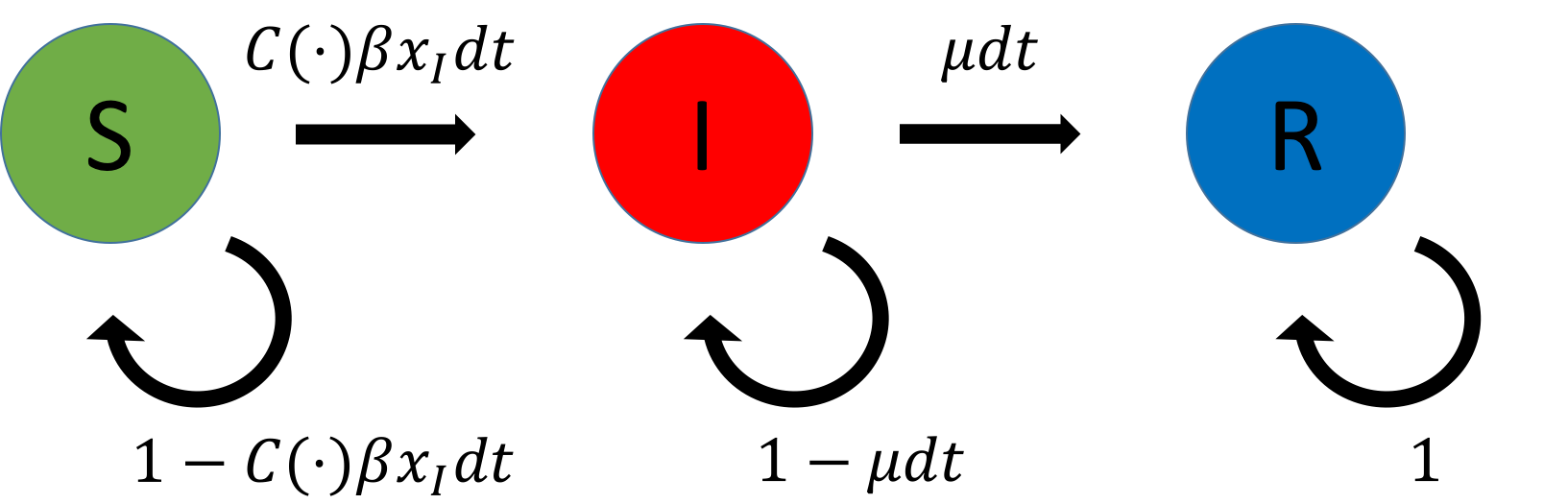} 
\caption{Between time $t$ and $t+dt$, S, I, and R individuals move between the states at these rates. The S individual becomes infected at a rate dependent on contact rates. The I individual recovers at a constant rate of $\mu dt$. The recovered individual remains in the state.}
\label{fig:prob_fig}
\end{figure}
\\
From these transition rates, we can write the Bellman equations, which give the value functions for individuals in S, I, and R states, respectively, as:
\begin{align}
V_S(t) &= \max_{c_S} \Big\{  \int_{t}^{t+dt} u_S(c_S) dt + (1 - C(\cdot)\beta x_I dt)V_S(t+dt) + C(\cdot)\beta x_Idt V_I(t+dt) \Big\} \nonumber\\
V_I(t) &= \max_{c_I} \Big\{  \int_{t}^{t+dt} u_I(c_I) dt + (1-\nu dt)V_I(t+dt) + \nu dt V_R(t+dt) \Big\} \nonumber\\
V_R(t) &= \max_{c_R} \Big\{  \int_{t}^{t+dt} u_R(c_R) dt + V_R(t+dt) \Big\} \label{eq:3}
\end{align}
with terminal conditions, $V_S(T)=V_I(T)=V_R(T)=0$. 
\subsection{Mean Field Equilibrium solution}
The SIR dynamics and the Bellman equations are coupled by the contact strategies, $c_S, c_I, c_R$, and the population, $x_S, x_I, x_R$. The solution to this problem is the mean-field equilibrium (MFE), which is the fixed point $(c_S^{eq}, c_I^{eq}, c_R^{eq}, x_S, x_I, x_R)$ such that 1) the strategies $c_S^{eq}(t)$, $c_I^{eq}(t)$, $c_R^{eq}(t)$ are the optimal solutions in equations \ref{eq:3} given $x_S, x_I,$ and $x_R$ and 2) $x_S$, $x_I$, and $x_R$ are solutions to the system of ODEs in equation \ref{eq:1} given the optimal strategies.
\subsection{Socially Optimal solution}
Additionally, we can characterize the socially optimal solution. Similar to the second-best equilibrium, discussed by Lipsey and Lancaster\cite{lipsey1956general}\cite{ohdaira2009cooperation}, if the infected individuals do not pursue the most utility-maximizing decisions, there is another equilibrium, in which the population can attain higher utility on average. If a central planner chooses the contact strategies of the population to maximize the utility of the entire population, the resulting socially optimal solution solves the dynamic optimization problem,
\begin{equation}
c_S^{opt}, c_I^{opt}, c_R^{opt} = \arg\max \int_0^T x_S(t)u_S(c_S(t)) + x_I(t)u_I(c_I(t)) + x_R(t) u_R(c_R(t))dt.
\end{equation}
We can also define modified version of the socially optimal problem, in which there is cost to move away from the MFE solution. Given the MFE solution, $(c_S^{eq},c_I^{eq}, c_R^{eq})$, the modified problem is 
\begin{align}
c_S^{opt}, c_I^{opt}, c_R^{opt} = \arg\max \int_0^T & \Big[ x_S(t)u_S(c_S(t)) + x_I(t)u_I(c_I(t)) + x_R(t) u_R(c_R(t)) \nonumber \\
&- \frac{1}{2}k \sum_z x_z(c_z - c_z^{eq})^2\Big]dt.
\end{align}
\section{Results}
\subsection{Mean Field Equilibrium solution}
The MFE solution is characterized through Proposition \ref{V_R}, \ref{V_I}, \ref{V_S}. 
\begin{prop} \label{V_R}
The optimal strategy for an R individual is $c_R^{eq}=0.5b_R$ and the corresponding optimal value function is $V_R(t) = -u_R^{max} t + u_R^{max} T$. 
\end{prop}
\begin{proof}
We substitute the Taylor expansion $V_R(t+dt) = V_R(t) + \Dot{V}_R(t)dt$, and take only the first order $dt$ terms.
\begin{align}
 V_R(t) &= \max_{c_R} \Big\{ \int_t^{t+dt} u_R(c_R)dt + V_R(t) + \Dot{V}_R(t) dt\Big\} \\
\Longrightarrow -\Dot{V}_R(t) &= \max_{c_R} \Big\{ u_R(c_R) \Big\}
\end{align}
$u_R$ attains its maximum, $u_R^{max}$, when $c_R^{eq} = 0.5b_R$, so we have
\begin{equation}
\Dot{V}_R(t) = -u_R^{max} \Longrightarrow V_R(t) = -u_R^{max}t + C
\end{equation}
Substituting the terminal condition, $V_R(T)=0$, we have $V_R(t) = -u_R^{max} t + u_R^{max} T$.
\end{proof}
In this model, the recovered individuals gain total immunity, so they always choose optimal contact rates.
\begin{prop} \label{V_I}
The optimal strategy for an individual in the infected state is $c_I^{eq}=0.5b_I$ and the corresponding optimal value function is $V_I(t) = V_R(T) - \frac{u_R^{max} - u_I^{max}}{\mu} (1-e^{\mu(t-T)})$
\end{prop}
\begin{proof}
Substituting the first order Taylor expansion and taking only the first order terms,
\begin{align}
& V_I(t) = \max_{c_I} \Big\{ \int_t^{t+dt} u_I(c_I)dt +(1-\mu dt)(V_I(t) + \Dot{V}_I(t)dt) + \mu dt (V_R(t) + \Dot{V}_R(t)dt) \Big\} \\
\Longrightarrow & V_I(t) = \max_{c_I} \Big\{ \int_t^{t+dt} u_I(c_I)dt + V_I(t) + \Dot{V}_I(t)dt -\mu V_I(t) dt + \mu V_R(t)dt \Big\} \\
\Longrightarrow & -\Dot{V}_I(t) = \max_{c_I} \Big\{ u_I(c_I) + \mu(V_R(t) - V_I(t))\Big\}
\end{align}
Because $u_I(c_I)$ is the only term dependent on $c_I$, the value function attains its maximum, $u_I^{max}$, when $c_I^{eq} = 0.5b_I$. 
\begin{equation}
-\Dot{V}_I(t) = u_I^{max} + \mu \Big( V_R(t) - V_I(t)\Big)
\end{equation}
$V_R(t)$ is known from Proposition \ref{V_R}, and this is a first-order linear ordinary differential equation that can be explicitly solved with integrating factor. Using the terminal condition, $V_I(T)=0$, we can find
\begin{equation}
V_I(t) = V_R(t) - \frac{u_R^{max} - u_I^{max}}{\mu} (1- e^{\mu (t-T)})
\end{equation}
\end{proof}
The I individuals, similar to the R individuals, do not change from the optimal contact rate. From the above formulation, we see that $V_I(t)$ is bounded by $V_R(t)$, and $V_R(t)-V_I(t)$ increases as $\mu$ decreases, since it implies that I individuals spend longer time as an infected before recovering.  
\begin{prop} \label{V_S}
The optimal strategy for a susceptible individual is $c_S^{eq}(t) < 0.5b_S$ during time $0\leq t < T$.
\end{prop}  
\begin{proof}
Substituting the Taylor expansion and taking only the first order terms, 
\begin{align}
& V_S(t)= \max_{c_S} \int_{t}^{t+dt}u_S(c_S)dt + \Big(1-c_Sc_I\beta x_I dt\Big)V_S(t+dt) + c_Sc_I \beta x_I dt V_I(t+dt) \\
\Longrightarrow & V_S(t) = \max_{c_S} \int_{t}^{t+dt} u_S(c_S)dt + V_S(t) + \Dot{V}_S(t) dt - c_Sc_I\beta x_Idt (V_S(t)-V_I(t)) \\
\Longrightarrow & -\Dot{V}_S(t) = \max_{c_S} u_S(c_S) - c_Sc_I\beta x_I (V_S(t) - V_I(t))
\end{align}
The objective function is concave, so we set the derivative equal to 0 to find $c_S^*$.
\begin{align}
\frac{du_S}{dc_S}\Big\rvert_{c_S=c_S^{eq}} = 0.5b_I \beta x_I (V_S(t) - V_I(t))
\end{align}
$u_S$ is concave, so $c_S^{eq}$ can be uniquely found. Also, we have $V_S(t) > V_I(t)$ for all $t < T$. This is because starting from the terminal condition $V_S(T) = V_I(T)=0$, if at any time $t$ we are sufficiently close to $V_S=V_I$, $\Dot{V}_S(t) = -u_S^{max} < -u_I^{max} = \Dot{V}_I(t)$. Therefore, $0.5b_I \beta x_I (V_S - V_I)>0$ for $t<T$, and $c_S^{eq}<0.5b_S$.
\end{proof}
Note that $c_S^{eq}$ is smaller for bigger values of $b_I \beta x_I (V_S(t)-V_I(t))$, which means that the susceptible individuals should decrease contact if 1) infected population gets large, 2) the disease spreads easily, 3) cost of being infected is large, or 4) if the disease minimally affects the ability of the infected.
\subsection*{Numerical solution}
We use discrete time steps $\Delta t$ to find the numerical solutions to the equations. With some initial $c^0= (c_S^0, c_I^0, c_R^0)$, we compute $x^0=(x_S^0, x_I^0, x_R^0)$ using the ODE forward equations. Then, $c^1$ can be computed via backward induction with the given $x^0$. We continue this until we find $c^k$ and $x^k$ such that $c^k=c^{k-1}$ and $x^k=x^{k-1}$, which is the MFE solution. \\
\begin{figure}[h!]
\centering
\includegraphics[scale=1]{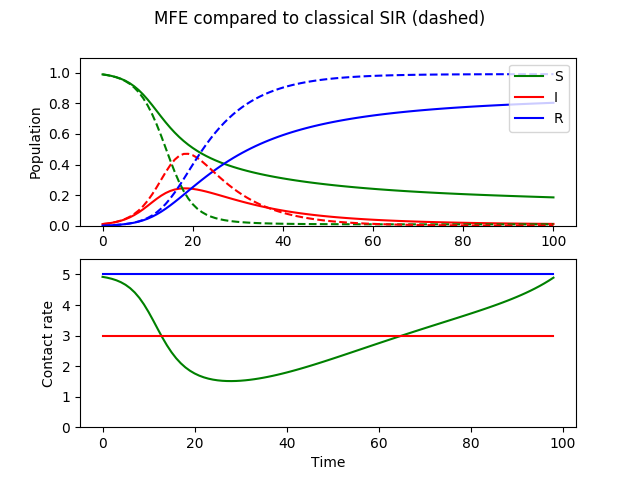} 
\caption{The top figure shows $x_z$ for $z \in \{S,I,R\}$, the SIR dynamics for MFE solution (solid lines) and the classical case without adaptive behaviors (dashed lines). Bottom figure shows $c_z^{eq}(t)$ for $z\in\{S,I,R\}$.}
\label{fig:mfe}
\end{figure}
\\
Figure \ref{fig:mfe} shows the MFE numerical solution with disease parameters $\beta=0.03$ and $\mu=0.1$, and  utility parameters $b_S=b_R=10$, $b_I=6$, $a_S=a_R=0$, $a_I=4$, and $\gamma=0.25$. As we showed in Proposition \ref{V_R} and \ref{V_I}, the recovered and infected individuals have no incentive to lower their contact rates, and so they continue with optimal level of contact rates. The susceptible individuals lower their contact rate to balance their immediate utilities and their expected cost of possibly getting infected. In result, the spread of infection is mitigated by the behavioral changes by the susceptible population, compared to the classical case in which no adaptive behavior is considered. The infection curve has lower peak due to the behavioral changes, but also a longer tail, because there are more remaining susceptible individuals to get the disease after the peak.\\ 
\begin{figure}[h!]
\centering
\begin{tabular}{cc}
\includegraphics[scale=0.5]{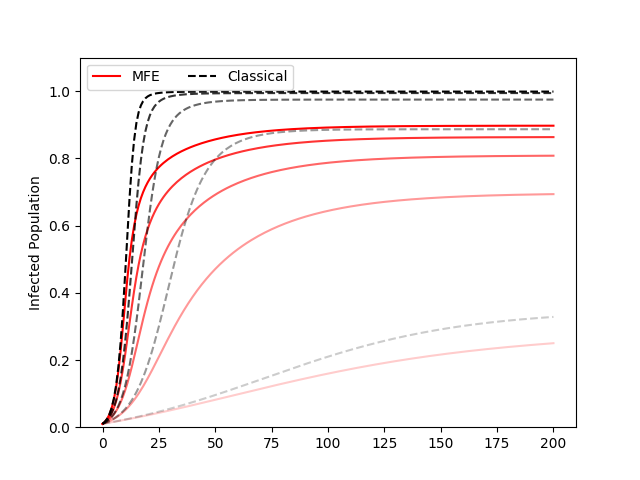} & \includegraphics[scale=0.5]{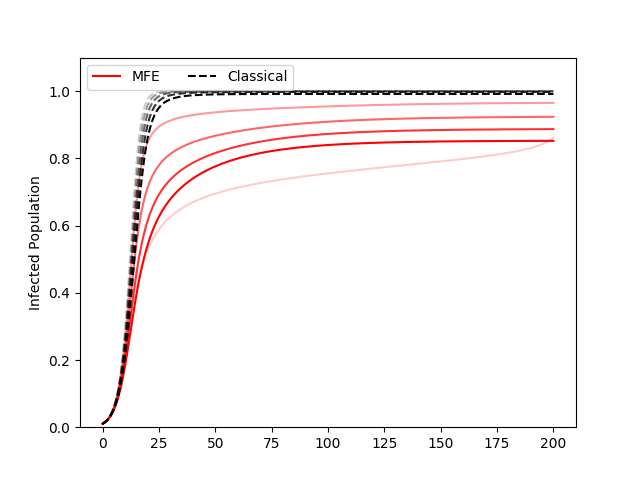} \\
(a) Cumulative curves $(0.01\leq \beta \leq 0.05)$ & (b) Cumulative curves $(0.02 \leq \mu \leq 0.1)$ \\
\includegraphics[scale=0.5]{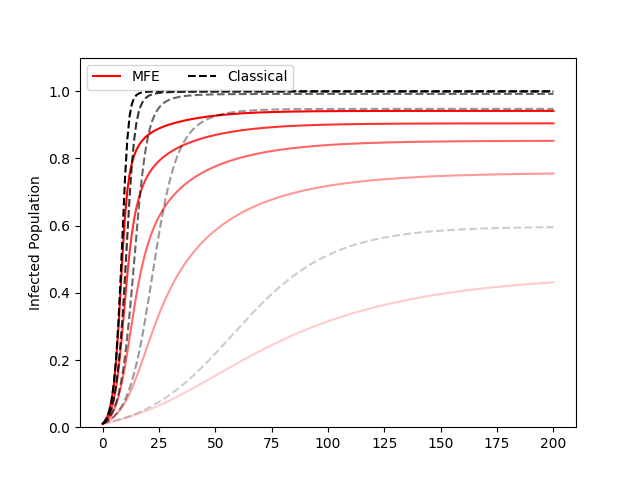} & \includegraphics[scale=0.5]{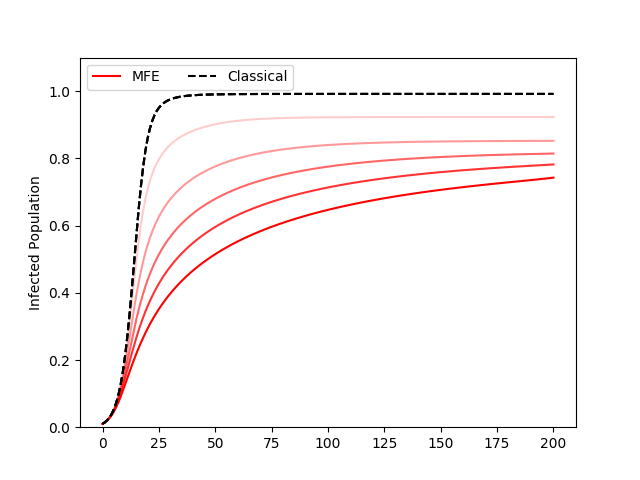} \\
(c) Cumulative curves $(2 \leq b_I \leq 10)$ & (d) Cumulative curves $(2 \leq a_I \leq 10)$
\end{tabular}
\caption{The cumulative epidemic size is shown for the MFE solution (red) and the classical SIR (black) for the range of given parameter. For each (a)-(d), five cumulative curves are shown corresponding to five values of the given parameter in the range, where the more transparent lines are smaller parameter values.}
\label{fig:fs}
\end{figure}
\\
Figure \ref{fig:fs} shows the cumulative epidemic curve of the MFE solution and the classical SIR for different parameter values of $\beta, \mu, b_I,$ and $a_I$. For any set of parameter values, the MFE solution, because of the adaptive behaviors, results in smaller final size as well as more gradual spread of the epidemic. \\
\\
If the disease is more infectious (large $\beta$), the MFE solution shows faster spread as well as larger final size, although mitigated compared to what the classical SIR model would predict (Figure \ref{fig:fs}a). If the disease has longer recovery time (small $\mu$), the classical SIR model predicts faster spread. However, the MFE solution predicts a different result. For example, when $\mu=0.02$, the classical case predicts that because the average infectious period is long, the number of infected population rises faster. However, the behavioral decision of individuals are considered in the MFE solution, and so a more gradual epidemic is predicted, since it is optimal to suppress contacts to avoid getting infected in the first place (Fig. \ref{fig:fs}b).\\
\\
$b_I$ and $a_I$ are new parameters introduced to consider the effect of behavioral change. $b_I$ denotes the economic utility, depending on social activity, and $a_I$ denotes the health utility of the individual. If the disease does not affect the productivity of the infected individual (high $b_I$), the disease will spread as if its transmission rate $\beta$ is higher (Fig. \ref{fig:fs}c). In the case of COVID-19, the symptoms of many have been mild to moderate, which might have contributed to the perceived higher transmission rate. \\
\\
The epidemic spread as predicted by the classical SIR model would not change with changes in $a_I$, but the MFE solution shows smaller epidemic and flatter growth rate for high $a_I$, as individuals choose to make less contacts to avoid the high cost of becoming infected (Fig. \ref{fig:fs}d). 
\subsection{Socially optimal solution} \label{ssec:so}
The MFE solution, discussed above, considers the case in which individuals maximize their own utility. For the socially optimal solution, we pose a centralized control problem, in which we find $(c_S^{opt}, c_I^{opt}, c_R^{opt})$ which maximizes the average utility of the entire population. Therefore we solve
\begin{align} \label{eq:19}
c_S^{opt}(t), c_I^{opt}(t), c_R^{opt}(t) &= \arg\max \int_0^T x_S(t)u_S(c_S(t)) + x_I(t)u_I(c_I(t)) + x_R(t)u_R(c_R(t)) dt \\
\text{subject to: }& \begin{cases}
\Dot{x}_S = -c_Sc_I\beta x_S x_I \\
\Dot{x}_I = c_Sc_I\beta x_S x_I - \nu x_I \\
\Dot{x}_R = \nu x_I \\
\end{cases}
\end{align}
In order to solve the optimal control problem, we use Pontryagin's maximum principle, which gives the necessary conditions for the optimal controls, given the evolving dynamics of the system.
\begin{theorem} [Pontryagin's Maximum Principle]
Let $\bm{x}=[x_S,x_I,x_R]^T$ and $\bm{c}=[c_S,c_I,c_R]^T$. For the given deterministic dynamics, $\Dot{\bm{x}}=f(\bm{x},\bm{c})$, the Hamiltonian is defined as 
\begin{equation*}
H(\bm{x}, \bm{c}, \bm{\lambda}, t) := L(\bm{x}, \bm{c}) + \bm{\lambda} ^T f(\bm{x},\bm{c})
\end{equation*} 
where $\bm{\lambda}(t)$ is the costate trajectory. If $\bm{x}(t)$, $\bm{c}^{opt}(t)$ is the optimal trajectory in $0\leq t \leq T$ from $\bm{x}(0)$, then $\bm{\lambda}(t)$ satisfies
\begin{equation*}
-\Dot{\bm{\lambda}} = H_{\bm{x}}(\bm{x}^{opt}, \bm{c}^{opt}, \bm{\lambda}, t) = L_{\bm{x}}(\bm{x}^{opt},\bm{c}^{opt}) + \bm{\lambda}^Tf_{\bm{x}}(\bm{x}^{opt},\bm{c}^{opt})
\end{equation*}
and $\bm{c}^{opt}$ is the solution to the optimization problem,
\begin{equation*}
\bm{c}^{opt} = \argmax_{\bm{c}} H(\bm{x}^{opt}, \bm{c}, \bm{\lambda})
\end{equation*}
\end{theorem}
\begin{prop}
The socially optimal contact rate of the infected, $c_I^{opt}$ must always be less than the MFE contact rate, $c_I^{eq}$ during time $0 \leq t < T$.
\end{prop}
\begin{proof}
We can apply the Pontryagin's maximum principle, which gives the necessary condition for optimality. If $c_S^{opt}, c_I^{opt}, c_R^{opt}$ are optimal solutions, then there exist Lagrangian multipliers, $\lambda_S(t), \lambda_I(t), \lambda_R(t)$, such that $\lambda_S(T)=0,\lambda_I(T)=0,\lambda_R(T)=0$, and for $t<T$, they satisfy:
\begin{align}
-\Dot{\lambda}_S &=c_S^{opt}c_I^{opt}\beta x_I (\lambda_I - \lambda_S) + u_S(c_S^{opt}) \\
-\Dot{\lambda}_I &=c_S^{opt}c_I^{opt}\beta x_S (\lambda_I - \lambda_S) + \mu(\lambda_R - \lambda_I) + u_I(c_I^{opt}) \label{eq:23}\\
-\Dot{\lambda}_R &= u_R(c_R^{opt}) \\
c_S^{opt},c_I^{opt},c_R^{opt}&=\arg\max c_Sc_I\beta x_Sx_I(\lambda_I-\lambda_S) + \mu x_I(\lambda_R-\lambda_I) + \sum_{z}x_zu_z(c_z)
\end{align}
Note the similarities between our expressions for $\lambda_S,\lambda_I,\lambda_R$ and $V_S,V_I,V_R$ from the MFE, with the additional term in (\ref{eq:23}) as the only difference. Intuitively, this additional term represents the infected individual caring about the consequences of its contact strategy on the population, which is not considered by the selfish I individual in its MFE objective function. It can be interpreted as the I individual's internalized negative externalities (i.e. thinking about our actions and their impact on others). First, we see that the $c_R$ term in the objective function can be separated, and so we can find the maximizer, $c_R^{opt} = 0.5b_R = c_R^{eq}$, which gives $\lambda_R(t)=V_R(t)$ for all $t$. Therefore, the optimization problem for $c_S^*$ and $c_I^*$ is 
\begin{equation}
c_S^{opt},c_I^{opt} = \arg \max  x_Su_S(c_S) + x_Iu_I(c_I) -  c_Sc_I\beta x_Sx_I(\lambda_S-\lambda_I) \label{eq:26}
\end{equation}
Clearly, $\beta x_S x_I > 0$ for all $t$, assuming nontrivial initial conditions, $x_S(0)>0$ and $x_I(0)>0$. Also, $\lambda_S - \lambda_I >0$ for $t<T$. Starting from the given terminal conditions $\lambda_S(T) = \lambda_I(T)=0$, if at any point we get sufficiently close to $\lambda_S =\lambda_I$, we see that $\Dot{\lambda}_S < \Dot{\lambda_I}$, so $\lambda_S$ will be strictly larger than  $\lambda_I$ in $0\leq t < T$. Therefore, $\beta x_Sx_I(\lambda_S - \lambda_I) > 0$ for all $0\leq t < T$. \\
\\
With the utility function $u_z(c_z)=(b_zc_z-c_z^2)^\gamma - a_z$, we will prove that $c_I^{opt} < 0.5b_I$ for all $t$ by dividing the problem into two cases: i) $\gamma < 1$ and ii) $\gamma=1$. \\
\\
i) $\gamma < 1$ \\
From (\ref{eq:26}), $c_I^{opt}=0.5b_I$ only if $c_S^{opt}\beta x_S x_I (\lambda_S - \lambda_I)=0$, which is only true if $c_S^{opt}=0$. However, we see that for some given $c_I^{opt}$, $c_S^{opt}$ can be computed to be
\begin{align}
& \frac{du_S}{dc_S}\Big\rvert_{c_S=c_S^{opt}} - c_I^{opt}\beta x_I(\lambda_S - \lambda_I) = 0 \\
\Longrightarrow & \frac{\gamma(b_S-2c_S^{opt})}{(b_Sc_S^{opt}-c_S^{opt2})^{1-\gamma}} = c_I^{opt} \beta x_I (\lambda_S - \lambda_I)
\end{align}
Since the left side monotonically decreases from $\infty$ to 0 in the domain $[0,0.5b_I]$, $c_S^{opt}$ can be uniquely found and since the right side is bounded, $c_S^{opt} > 0$. Therefore, $c_I^{opt} < 0.5b_I$ must be true. \\
\\
ii) $\gamma = 1$ \\
Plugging in $\gamma=1$ gives the optimization problem as
\begin{equation}
c_S^{opt},c_I^{opt} = \arg\max x_S (b_Sc_S - c_S^2) + x_I (b_Ic_I - c_I^2) - c_Sc_I \beta x_S x_I (\lambda_S - \lambda_I)
\end{equation}
Assume for contradiction that $c_I^{opt}=0.5b_I$ at some $t < T$. This is only possible if $c_S^{opt} = 0$, which implies that $b_S \leq c_I^{opt}\beta x_I (\lambda_S - \lambda_I)$ holds true. Since the optimal contact rates at time $t$ is given as $c_S^{opt}(t)=0$ and $c_I^{opt}(t)=0.5b_I$, we can plug these in to get
\begin{align}
\Dot{\lambda}_S(t) &= 0 \\
\Dot{\lambda}_I(t) &= -\mu (\lambda_R(t) - \lambda_I(t)) - u_I^{max}
\end{align}
With the known dynamics of $\lambda_S,\lambda_I,$ and $x_I$, we can use the first order Taylor expansion to find
\begin{align}
x_I(t+dt) &\Big(\lambda_S(t+dt) - \lambda_I(t+dt)\Big) = (x_I -\mu x_Idt) \Big( \lambda_S - \lambda_I + \mu(\lambda_R - \lambda_I)dt + u_I^{max}dt\Big) \\
&=x_I \Big( \lambda_S - \lambda_I + \mu(\lambda_R - \lambda_I)dt + u_I^{max}dt - \mu(\lambda_S - \lambda_I)dt \Big) > x_I(\lambda_S - \lambda_I)
\end{align}
where the last inequality is due to $\lambda_R$ being the upper bound of $\lambda_S$. We see that $c_I^{opt} \beta x_I (\lambda_S - \lambda_I)$ is monotonically increasing in time, and therefore, if $(c_S^{opt},c_I^{opt}) = (0,0.5b_I)$ is the optimal solution at time $t$, it is also the solution at time $t+dt$. This means that $\lambda_S - \lambda_I$ must be increasing in time, but then it cannot satisfy the transversality condition, $\lambda_S(T) = \lambda_I(T) = 0$. Therefore, by contradiction, $c_I^* < 0.5b_I$ for all $0\leq t < T$.
\end{proof}
This proposition implies that the selfish behavior of the infected is never optimal for the population, and so the appropriate policy in all cases is to decrease the contacts of the infected below the level that they want to selfishly make.
\subsection*{Numerical Solution}
For the same set of parameters as in Figure \ref{fig:mfe}, we can compute the socially optimal solution.
\begin{figure}[h!]
\centering
\begin{tabular}{cc}
\includegraphics[scale=0.5]{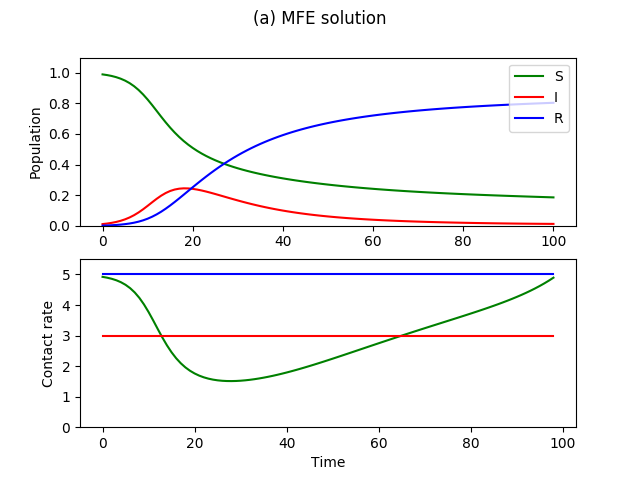} & \includegraphics[scale=0.5]{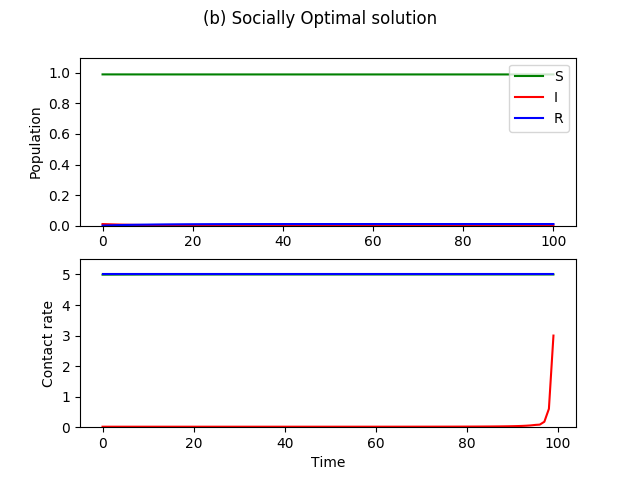} \\
\includegraphics[scale=0.5]{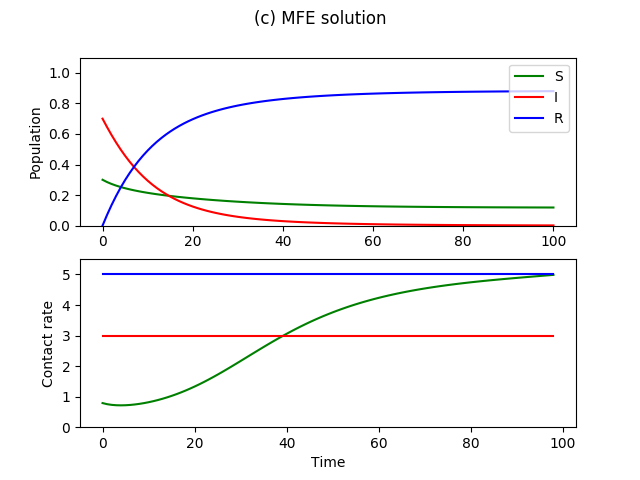} & \includegraphics[scale=0.5]{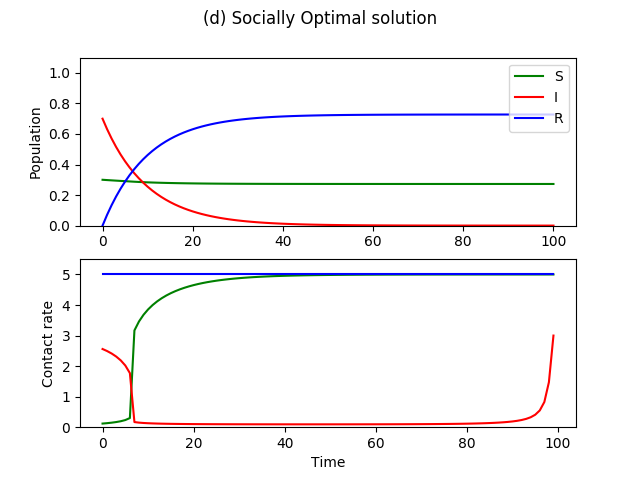} \\
\end{tabular}
\caption{(a) shows the same MFE solution as in Figure \ref{fig:mfe}. (b) shows the socially optimal solution, which maximizes the total utility of the population. (c) and (d) show the MFE solution and socially optimal solution for a different initial population, ($x_S(0),x_I(0))=(0.3,0.7)$.}
\label{fig:mfeso}
\end{figure}
The optimal solution as shown in Figure \ref{fig:mfeso}b is to completely suppress the contacts of the I, so that no additional infection can take place, while the susceptible population resumes normal activities. In fact, this is most often the socially optimal strategy. An exception is if a large part of the population was already infected initially (Figure \ref{fig:mfeso}c,d). Here, 70 percent of the population is initially infected, and so the optimal contact strategy is to completely isolate the susceptible until the number of infected decreases to some level, at which, we go back to complete quarantining of the I and the resuming of normal activities by the S. This type of policy is the complete shutdown implemented to various degrees in the United States, in which it is better to completely insulate the susceptibles until the number of infected population decrease back to a manageable level.  
\subsection{Cost of central planning}
A key piece missing in our characterization of the socially optimal solution is the various costs to controlling the contact rate of the population. The assumption in Section \ref{ssec:so} that the central planner can freely choose contact rates is far from realistic. For example, in the optimal solutions of Figure \ref{fig:mfeso}b, complete suppression of contact rates of the infected must be done through quarantining or providing proper incentives to keep infected people from making social contacts, which all cost resources. Additionally, people do not generally like their choices to be decided by an authority, so this cost also includes the "loss in freedom," which the central planner of wants to minimize. Therefore, we add to the problem in \ref{ssec:so}, a cost for deviating away from the population's MFE. 
\begin{align}
c_S^{opt}(t), c_I^{opt}(t), c_R^{opt}(t) = \argmax_{\bm{c}} \int_0^T & \Big[ x_S(t)u_S(c_S(t)) + x_I(t)u_I(c_I(t)) + x_R(t)u_R(c_R(t)) \nonumber \\
&- \frac{1}{2}k\sum x_z(c_z-c_z^{eq})^2 \Big] dt
\end{align}
We add a new parameter, $k \geq 0$, which balances the competing objectives of maximizing total utility and minimizing the deviations from selfish strategies, similar to a regularization parameter in machine learning. When $k=0$, the solution to the objective function is the socially optimal solution, and when $k$ is large, the solution is the MFE solution. In the intermediate values of $k$, we can find the solution which balances the trade-off. \\
\begin{figure}[h!]
\centering
\begin{tabular}{ccc}
\includegraphics[scale=0.6]{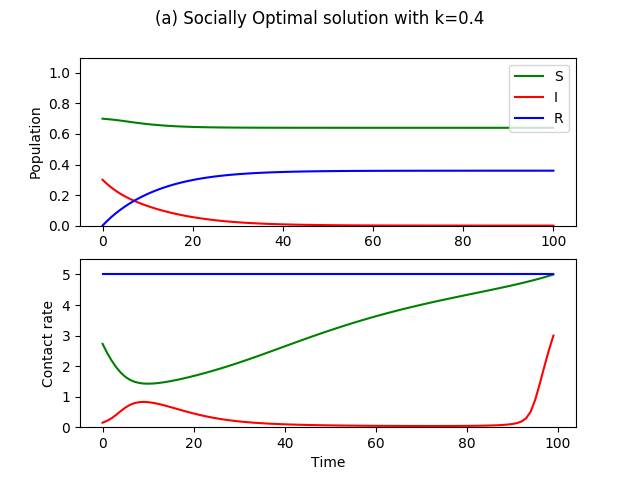} \\ \includegraphics[scale=0.6]{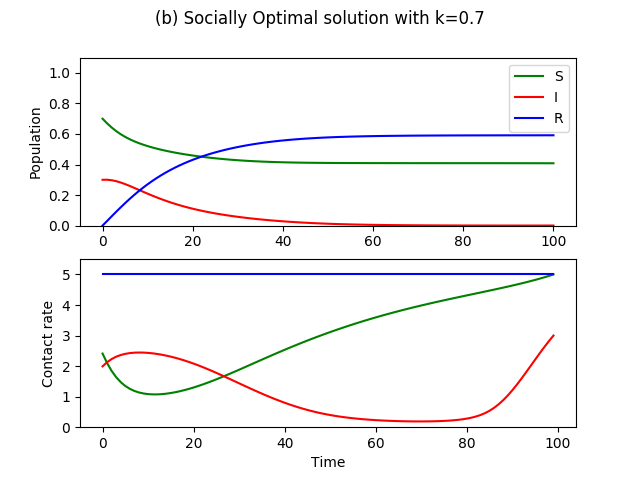} \\
\includegraphics[scale=0.6]{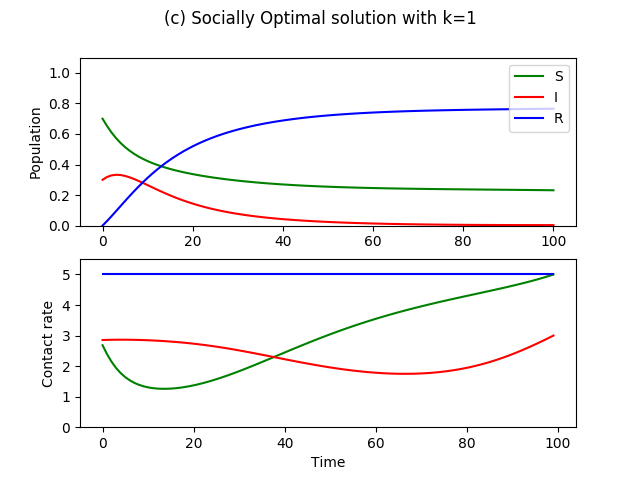} 
\end{tabular}
\caption{With initial condition $x_I(0)=0.3$, we compute the socially optimal contact strategies and corresponding population dynamics for $k=0.4,0.7$, and $1$. As $k$ is increased, it becomes more expensive to shift the infected strategies from the selfish strategy, and so the socially optimal contacts of the infected is larger, as the cost of lowering it begins to outweigh the utility benefits.}
\label{fig:socost}
\end{figure}
\\
As we see in Figure \ref{fig:socost}, for different values of $k$, we compute the socially optimal contact rates of the population. Another interpretation is that given limited resources to control contact rates, we find how it should be distributed during $t \in [0,T]$. We see a common result among the range of values for $k$, which is that when susceptibles have low contact rates, it is less important to keep the contact rates of the infected as low. Instead, it is more optimal use of resources to make sure the contact rates of the infected after the peak of the epidemic is lowered because this is when the susceptibles start ramping up to normal activities, thus posing higher risk of a second epidemic. Another reason is that because the number of infected is smaller at this time and the imposed cost is fixed per capita, the same amount of resources is more efficiently used by controlling the smaller population of infected rather than the larger infected population during the peak.
\subsection{Application to SEIR Model} \label{seir}
Here we show the numerical solutions to the MFE problem when exposed class (E) is added to the model. First, we consider the case in which the exposed individuals are not yet infectious, but incubating the disease. Second, we consider the case in which the exposed individuals are presymptomatic, meaning that they are infectious before they show symptoms, although at a lower rate compared to the infectious individuals. \\
\\
The incubation period is a random variable which has exponential distribution with parameter $\nu$ (The average incubation period is $\nu^{-1}$). Another variable which is introduced is $\beta_E$, the presymptomatic transmission rate of individuals in the exposed class. If there is no presymptomatic transmission, $\beta_E=0$. The SEIR model is
\begin{align}[left=\empheqlbrace] 
\Dot{x}_S &= -C(\cdot)\beta x_Sx_I - C(\cdot)\beta_E x_Sx_E \nonumber\\
\Dot{x}_E &= C(\cdot)\beta x_S x_I + C(\cdot)\beta_E x_Sx_E - \nu x_E \nonumber\\
\Dot{x}_I &= \nu x_E - \mu x_I \nonumber\\
\Dot{x}_R &= \mu x_I \label{eq:1}
\end{align}
Because the susceptible and the exposed classes do not have visible symptoms, an individual cannot differentiate between being in the S or E class. Therefore, we define the asymptomatic class, which includes both the susceptible class and the exposed class. The value function of this class will be taken as the weighted average between the S and E classes,
\begin{equation}
V_A = \frac{x_S V_S + x_E V_E}{x_S + x_E}
\end{equation}
where the corresponding value functions are given as
\begin{align}
V_S(t) &= \max_{c_S} \Big\{  \int_{t}^{t+dt} u_S(c_S) dt + \Big(1 - C(\cdot)\beta x_I dt - C(\cdot)\beta_E x_Edt\Big)V_S(t+dt)\\
&\hspace{25mm}+\Big(C(\cdot)\beta x_Idt + C(\cdot)\beta_E x_Edt \Big)V_E(t+dt) \Big\} \nonumber\\
V_E(t) &= \max_{c_E} \Big\{  \int_{t}^{t+dt} u_S(c_E) dt + (1 - \nu dt)V_E(t+dt) + \nu dt V_I(t+dt) \Big\} \nonumber\\
V_I(t) &= \max_{c_I} \Big\{  \int_{t}^{t+dt} u_I(c_I) dt + (1-\mu dt)V_I(t+dt) + \mu dt V_R(t+dt) \Big\} \nonumber\\
V_R(t) &= \max_{c_R} \Big\{  \int_{t}^{t+dt} u_R(c_R) dt + V_R(t+dt) \Big\} \label{eq:3}
\end{align}
Therefore the optimal contact rates chosen by the asymptomatic class is
\begin{equation}
V_A(t) = \max_{c_A} \Big\{ \frac{x_S L_S(c_A) + x_E L_E(c_A)}{x_S + x_E} \Big\}
\end{equation}
where $L_S(\cdot)$ and $L_E(\cdot)$ are the objective functions for the $S$ and $E$ classes in equations \ref{eq:3}.\\
\\
\begin{figure}[h!]
\centering
\begin{tabular}{cc}
\includegraphics[scale=0.5]{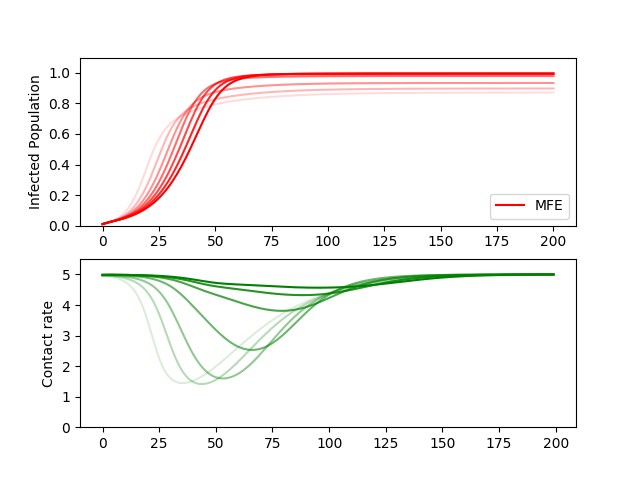} & \includegraphics[scale=0.5]{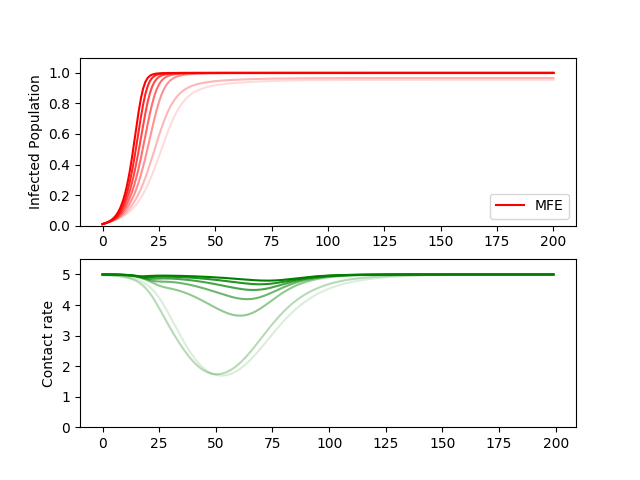} \\
(a) $\nu^{-1} \in [2, 15]$ & (b) $\beta_E \in [\frac{\beta}{16}, \frac{\beta}{2}]$
\end{tabular}
\caption{(a) The average incubation time period, $\nu^{-1}$ is varied from 2 to 15 days, and the corresponding infection curve and the optimal contact decisions of the susceptibles are shown from lighter to darker curves. (b) The presymptomatic transmission rate as a fraction of $\beta$, the regular symptomatic transmission rate is varied from $\frac{1}{16}$ to $\frac{1}{2}$ and the corresponding infection curve and the optimal contact decisions of the susceptibles are shown from lighter to darker curves.} 
\label{fig:exposed}
\end{figure}
With the parameters from Figure \ref{fig:mfe}, we include an exposed class and vary the incubation parameter $\nu$, where $\nu^{-1}$ is the average incubation period (Figure \ref{fig:exposed}a) or the presymptomatic transmission rate $\beta_E$ (Figure \ref{fig:exposed}b. \\
\\
As seen in Figure \ref{fig:exposed}a, if there is an epidemic of an infectious disease with longer incubation period, it becomes optimal for susceptible individuals to not engage in social distancing. This is because larger incubation period suggests higher uncertainty of whether one was already exposed or not. The infection curve also reaches a higher final size because of more contact behavior, although at a slower rate because of the incubation period. \\
\\
Additionally, if the infectious disease causes presymptomatic transmission, the resulting infection curve and optimal contact decisions of the susceptibles can be seen in Figure \ref{fig:exposed}b. If $\beta_E$, the presymptomatic transmission rate, is large, it quickly promotes more behavior because individuals of the exposed class unknowingly make social contacts, and quickly infect many susceptibles. \\
\\
Here, we presented preliminary results of this model when an exposed class is included. Further studies are needed such as computing the socially optimal solution with and without cost, as we did for the SIR model. 
\subsection{Price of Anarchy}
We compute the price of anarchy (PoA), which is a measure of how much the system degrades due to the selfish strategies of individuals in each S, I, R compartment. In the context of mean-field games, it is the ratio of the total utility of the population adopting MFE strategies to the total utility of the population adopting socially optimal strategies\cite{carmona2019price}. The PoA is given by
\begin{equation}
\text{PoA} = \frac{V_{opt}(0)}{x_S(0)V_S(0) + x_I(0)V_I(0) + x_R(0)V_R(0)}
\end{equation}
where $V_{opt}$ is the maximum of the objective function in (\ref{eq:19}). 
\begin{figure}[h!]
\centering
\begin{tabular}{cc}
\includegraphics[scale=0.5]{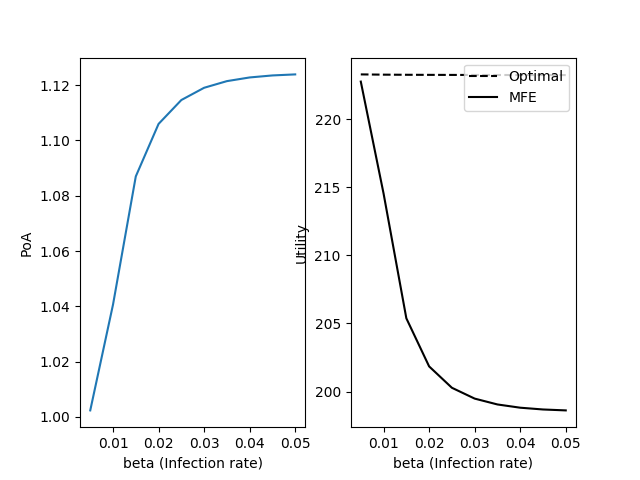} & \includegraphics[scale=0.5]{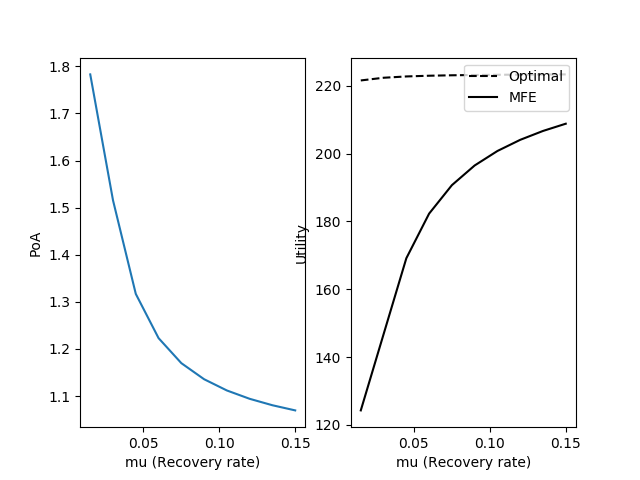} \\
(a) PoA of $\beta \in [0.005,0.05]$ & (b) PoA of $\mu \in [0.015, 0.15]$ \\
\includegraphics[scale=0.5]{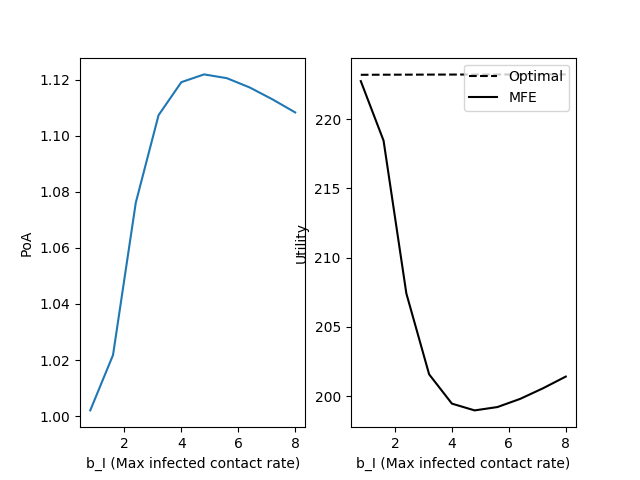} & \includegraphics[scale=0.5]{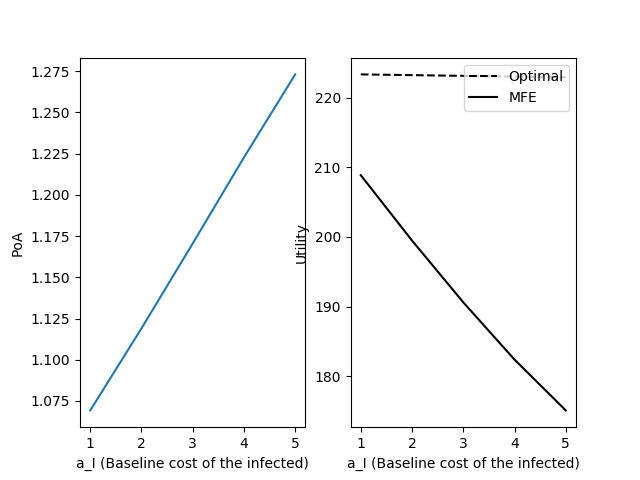} \\
(c) PoA of $b_I \in [1,8]$ & (d) PoA of $a_I \in [1, 5]$
\end{tabular}
\caption{We vary the model parameters and compute the price of anarchy (left) and the total population utility of the socially optimal strategy (right, dashed) and the MFE strategy (right, solid). The parameter space where PoA is high is where intervening public policy would be most needed, since this is the case in which selfish strategies are most degrading the total utility.}
\label{fig:poa}
\end{figure}
The PoA, computed for ranges of $\beta, \mu, b_I,$ and $a_I$, is shown in Fig. \ref{fig:poa}. While one parameter is changed, the others were kept fixed with the parameter values from Fig. \ref{fig:mfe}. \\
\\
As $\beta$ is made larger, PoA increases because each selfish contact behavior of the infected individual becomes magnified by the high transmission rate per contact (Fig. \ref{fig:poa}a). We see a similarly increasing trend as $\mu$ gets smaller, because it increases the time spent infected (Fig. \ref{fig:poa}b). However, an increase in $\beta$  results in decreasing marginal gain in PoA while a decrease in $\mu$ results in increasing marginal gain in PoA. Even for unreasonably large $\beta$, the overall utility suffers by around 11\%, while small $\mu$ can cause a decline of 45\%. The PoA is affected more by changes in $\mu$, the recovery rate of the disease.\\
\\
Fig. \ref{fig:poa}c shows that an intermediate value of $b_I$ results in the largest PoA. This non-monotonic relationship is because of the trade-offs of large $b_I$. On one hand, large $b_I$ means that the disease does not affect the day-to-day productivity of the infected as severely, and so the infected are not as penalized. On the other hand, this also means that the infected are able to be more active and make more contacts, which infects more susceptibles. These two opposing effects are balanced near $b_I=5$, where the PoA is at its maximum. When $a_I$ is large (Fig. \ref{fig:poa}d), the baseline cost of getting infected is larger, which results in less utility at each time point. 
\section{Discussion}
\subsection{Possible additions}
The public policy response to COVID-19 is an extremely complex problem with many factors which this paper has not covered. This model and analysis are, in many ways, the simplest baseline case from which we can make more realistic to fit a particular disease. First, we can pose the problem with different compartmental models. For example, if we take the SIS model, we would see different optimal strategies since the infected also face the burden of social distancing. In the case of COVID-19, it will be most useful to include a compartment with asymptomatic transmission, which behaves like S from the central planner's perspective. The socially optimal strategies depended on being able to distinguish between the S and the I, when it is not always the case. The lack of available testing of COVID-19, for example,  provides the uncertainty within the population as well as from the central planner's perspective. Second, we can add heterogeneity to the population by including additional state variables such as age, socioeconomic status, or level of prosociality. By explicitly adding the different subpopulations, we can understand the game theoretic dilemma at play between the old vs. the young, the financially stable vs. the unstable, or the prosocial vs. antisocial. Then, more specific policies may be proposed that target the contact strategies of a particular group.
\subsection{Conclusions from this model}
While our analysis does not include many important factors, we can still make some general conclusions.
\subsubsection*{Selfish strategies still "flatten the curve."}
By including adaptive behavior of individuals, our model predicts epidemic curves with flatter growth rate, compared to the classical counterparts. In our simple model, the curve is flattened naturally because of susceptibles who weigh the trade-off between current utility of making contacts and the future cost of getting infected. The recovered and the infected do not have any trade-offs to decrease contacts. Even if only the susceptibles are practicing social distancing, it still decreases the number of contacts of the system, and so the infected population reaches a smaller peak (Fig. \ref{fig:mfe}, \ref{fig:fs}). It should be emphasized that the curve is flattened because individuals anticipate future growth in infections and decrease their contacts to avoid being exposed to the infected individuals. If the possible outbreak is flat-out denied by the media, then individuals will not adapt their behaviors, causing an unmitigated large peak in the infected population. Therefore, it is the responsibility of the policymakers to clearly communicate the existence and extent of the spreading disease.
\subsubsection*{Selfish strategy of the infected is never socially optimal.}
We prove that the socially optimal strategy of the infected, $c_I^{opt}$, is always less than $c_I^{eq}$. Therefore, policies in response to the epidemic should decrease the contact rates of the infected. We see examples of policies with this aim such as quarantining the infected or granting paid sick leave to individuals who tested positive. Both policies respectively decrease the contact rates of the infected directly or indirectly by decreasing $b_I$, reducing the potential gain in utility of making more contacts. Because reducing the contacts of the infected is so important, policymakers might consider even more aggressive policies. 
\subsubsection*{It is important to control the infected contact rates, following the peak of the epidemic.}
If cost is imposed to the central planner in changing the contact rates of the individuals, we find the new socially optimal contact rates, depending on $k$, which is the per capita unit cost of changing contact rates. For $k=0$ and $k \gg 1$ respectively, we find the cost-imposed socially optimal solution to be the previously computed $\bm{c}^{opt}$ and $\bm{c}^{eq}$. For $k$ values in between (Fig. \ref{fig:poa}), we commonly see that when it is too costly to decrease the infected contact rate for the whole time period, it is most beneficial to at least focus on decreasing after the epidemic has subsided. An assumption here is that cost of central planning is constant in time, when it may not be in real world situations. When outbreak is at its peak, more public attention is on the disease, and it may be easier to implement social distancing or secure funding for quarantining. However, when the disease has subsided, it might be harder to convince the public to behave differently.\\
\\
This result reinforces the need for formal social distancing policy which goes beyond the peak of the epidemic. When the disease is prevalent, social distancing can be naturally favored due to individual optimization, but to sustain it for longer requires centralized public policy to prevent second peaks. This general result is in agreement with other studies of COVID-19 policies which mention the likely possibility of second peaks\cite{di2020timing,morris2020optimal,ferguson2020report}. Additional work, using more realistic central planning cost depending on time and population structure, will help us better understand how such long-term social distancing policies should be implemented.
\subsubsection*{Disease with exposed period is unfavorable for promoting social distancing.}
As the preliminary results show in Section \ref{seir}, if an epidemic of an infectious disease with an incubation period breaks out in a population, it will be harder to maintain social distancing behaviors in the population. In the case of COVID-19, which is believed to have an incubation period of up to two weeks before symptoms, each asymptomatic individual will be more uncertain concerning its infection status, and thus make it harder to prolong social distancing. Additionally, it is believed that those exposed to COVID-19 can transmit the disease before showing symptoms, which will further make it hard to maintain social distancing, as seen in Figure \ref{fig:exposed}b.
\subsubsection*{Policies are most needed for diseases with low $\mu$, high $a_I$, high $\beta$ and intermediate $b_I$.}
By computing the price of anarchy, we can measure the effect of different parameters on how much the system is degraded by the selfish behaviors. Diseases with low $\mu$, high $a_I$, high $\beta$, intermediate $b_I$, in this order, seem to most affect the population such that their selfish behaviors will degrade the system more compared to central intervention.\\
\\
An interesting future work will be to put different diseases on the spectrum of these 4 variables, depending on its epidemiological characteristics as well as its economic and health effects on the infected. Then, we can roughly categorize diseases which need to be centrally intervened in the case of an outbreak.

\bibliographystyle{unsrt}
\bibliography{sir_mfg.bib}
\end{document}